\documentclass[pdflatex,sn-basic,iicol]{sn-jnl}% Math and Physical Sciences Author Year Reference Style
\usepackage{graphicx}%
\usepackage{multirow}%
\usepackage{amsmath,amssymb,amsfonts}%
\usepackage{amsthm}%
\usepackage{mathrsfs}%
\usepackage[title]{appendix}%
\usepackage{xcolor}%
\usepackage{textcomp}%
\usepackage{manyfoot}%
\usepackage{booktabs}%
\usepackage{algorithm}%
\usepackage{algorithmicx}%
\usepackage{algpseudocode}%
\usepackage{listings}%
\let\orcidlogo\relax

\usepackage{orcidlink}
\theoremstyle{thmstyleone}%
\newtheorem{theorem}{Theorem}%  meant for continuous numbers
 
\newtheorem{corollary}{Corollary}[theorem]

\newtheorem{proposition}[theorem]{Proposition}% 

\theoremstyle{thmstyletwo}%
\newtheorem{example}{Example}%

\theoremstyle{thmstylethree}%
\newtheorem{definition}{Definition}%

\DeclareMathOperator*{\argmax}{\arg\max}

\newcommand{\N}{\mathcal{N}}
\newcommand{\remove}[1]{}
\newcommand{\cP}{\mathcal{P}}
\newcommand{\bp}{\mathbf{p}}
\newcommand{\bq}{\mathbf{q}}
\newcommand{\bu}{\mathbf{u}}
\newcommand{\x}{\mathbf{x}}
\newcommand{\y}{\mathbf{y}}

\begin{document}

\title[Article Title]{An Information Theoretic Treatment of Yager's
Probability Distribution  Negation}


\author{\fnm{Roberto} \sur{Bruno}}\email{rbruno@unisa.it}

\author{\fnm{Ugo} \sur{Vaccaro}}\email{uvaccaro@unisa.it}

\affil{\orgdiv{Department of Computer Science}, \orgname{University of Salerno}, \orgaddress{\street{Via Giovanni Paolo II, 132}, \city{Fisciano}, \postcode{84084}, \state{Salerno}, \country{Italy}}}


\abstract{
{
In the seminal paper \citep{Yager2015}, Yager defined the negation  of a 
probability distribution $\bp=(p_1,\dots,p_n)$, as the distribution 
 $\overline{\bp} = (\overline{p}_1,\dots,\overline{p}_n)$, where 
    $\overline{p}_i = ({1-p_i})/({n-1}),$ for $ i=1, \ldots , n.$

}

In this paper, we present a comprehensive information-theoretic analysis of Yager’s negation and its generalizations.
%related classes of probability distribution negations, with particular emphasis on the class of independent negators. 
Using tools from information theory and majorization theory, we 
{unify, extend, and strengthen a number of previously known properties of
Yager's negation within a common framework.}

Overall, our results
%provide new theoretical insights into the structure and properties of probability distribution negations and 
offer strong theoretical justification for Yager's negation as the most natural and principled definition of probability distribution negation under various 
{information theoretic} criteria.}


\maketitle
\section{Introduction}
The concept of negation is fundamental to human reasoning and communication, serving as a primary mechanism to express opposition, denial, or the complement of a concept. In classical logic, negation is well-defined and unambiguous: it relates a proposition $P$ to a proposition $\neg P$ (``not $P$") such that if $P$ is true, 
then $\neg P$ is false, and vice versa. This framework extends naturally to fuzzy logic \citep{Zadeh1965}, where the truth value of a proposition can be any real number $a \in [0, 1]$, and its negation is 
given a truth value equal to $1-a$. However, when knowledge is represented by probability distributions over a set of multiple outcomes, defining ``negation" becomes significantly more challenging, {since many mathematically consistent definitions of negation are possible, as 
we shall see. One of the purposes of this paper is  to show that the definition 
proposed in \citep{Yager2015} is, from an \lq\lq axiomatic" 
Information Theoretic point of view, the right concept of probability distribution negation.}

\subsection{Motivating Scenarios}\label{sec:motivations}

The necessity to negate probability distributions is not merely a theoretical curiosity; it arises in artificial intelligence and knowledge representation systems. For instance, if an intelligent system models ``High Price" as a probability distribution, it must also model ``Not High Price" to process negative inference rules (e.g., IF Price is Not High THEN Buy) \citep{generating,kreinovich2018beyond}. 

Another example arises in the domain of target recognition and sensor data fusion, where
the reliability of information fusion is frequently compromised by high-conflict evidence, that is,
scenarios where disparate sensors report mutually exclusive hypotheses with high confidence. Standard fusion rules, such as Dempster’s classical combination rule \citep{Dempster1967}, often fail under these conditions, leading to counterintuitive results. To address this problem, \cite{Gao+} proposed using generalized negations to reduce conflict among divergent evidence sources before combining them. Their work showed that in noisy environments, negative evidence (e.g., ruling out a target) is frequently more reliable than positive identification. Similarly, \cite{Xu+} showed that negating probability distributions enables mass redistribution that mitigates evidential conflict, providing a more robust framework for target recognition under noisy and contradictory conditions.

The negation of probability distributions has also been proved 
to be valuable in Multi-Criteria Decision Making. In this area, the main challenge lies in deciding how much ``weight" or importance one has to assign to some criteria when the underlying data is unreliable. \cite{Sun2020} utilized negation to address this problem by measuring the \lq\lq distance"
between supporting evidence of some criteria,
and its negation. If the supporting
evidence is very different from its negation, then the authors of \citep{Sun2020} assumed that said evidence be 
reliable and therefore is given more weight. Vice versa, if it is too similar to its negation, then such
evidence  is vague or uncertain; thereby, it is given a smaller weight. As another example, \cite{tanwar2023generalization} proposed a generalization of negation based on biased distributions and applied it to medical diagnosis, demonstrating its effectiveness in distinguishing between ailments that share a common set of symptoms.

Quite interestingly, the concept of probability negation also emerges in modern machine learning, specifically within the paradigm of learning from complementary labels \citep{ishida2017learning}. In this framework, instead of providing the true label for each data point in the training data, one provides a ``complementary'' label: a single class to which the observation does \textit{not} belong. Then, to train a model using these complementary labels, the target probability distribution is constructed by distributing the probability mass equally among all remaining alternatives. Mathematically, this uniform reallocation is exactly equivalent to the formulation of negation proposed by \cite{Yager2015}, which will be discussed later in this section. Moreover, when this complementary label approach is combined with regularization techniques such as label smoothing \citep{szegedy2016rethinking} to mitigate model overfitting, the resulting target distribution takes the exact mathematical form of the broader class of negations introduced in \citep{generating, klein2022some}.

The varieties of applications of the negation of probability distributions 
naturally lead to the following fundamental problem: how can we put forward a principled definition of the 
 negation of a probability distribution?  {Since probabilities must sum to one, negating
 the probability $p_i$ of an event $x_i$ inherently shifts
 probability mass towards its complement, 
 requiring a redistribution 
 of mass among the remaining alternative outcomes.}
 By leveraging the Dempster-Shafer theory \citep{Shafer1976, smets1994,smets1994_2}, in the seminal paper \citep{Yager2015}, Yager proposed to
define the negation of a probability distribution
on the basis of the redistribution of probability masses according to a maximum entropy allocation principle. He argued that in the absence of specific information favoring one alternative over another, the probability mass of the negated event should be redistributed equally among the other outcomes to preserve neutrality.
The entropy measure Yager used is the so-called \lq\lq logical entropy" (see
\citep{ell} and references therein quoted). Formally, for a probability distribution $\bp=(p_1,\dots,p_n)$,  \cite{Yager2015} defined its 
negation $\overline{\bp} = (\overline{p}_1,\dots,\overline{p}_n)$ as:
\begin{equation*}
    \overline{p}_i = \frac{1-p_i}{n-1}, \quad i=1, \ldots , n.
\end{equation*}

This transformation ensures that the negation process introduces the maximum amount of uncertainty regarding the remaining outcomes, thereby avoiding unjustified bias.

Since Yager's proposal, 
many variants and generalizations of his approach 
have been proposed. \cite{generating} generalized the concept by viewing negations as point-by-point transformations of the original probabilities,  using appropriate decreasing functions called ``negators".
\cite{TANWAR2023113557} argued that equal redistribution is not always ideal, suggesting
a  framework with an unequal redistribution of probabilities, which generalizes Yager's negation. Over the years, various other notions of negation have been defined and proposed depending on the context; for further details, see \citep{ Kaur_comment,Kaur_markov, kaur, LiuDL, Luo_matrix2020, pham_estimating2021, Wu_exp,Xu+,YinDD} and the references quoted therein. 

In parallel, researchers have also investigated the theoretical properties of these transformations that go 
from probabilities distributions to their negations, e.g., \citep{klein2022some,TANWAR2023113557,Yager2015}. 
The present paper aligns directly with this line of research.

\subsection{Previous Work}

As Yager's paper \citep{Yager2015} has spanned a sizeable
literature, in this section we limit ourselves to discussing the work that more directly
investigates
the consequences and the mathematical properties of Yager's probability negation,
and its generalizations.

Yager himself \citep{Yager2015} observed that repeated applications of Yager’s negation progressively increase the uncertainty of the distribution (i.e., its entropy), eventually converging to the uniform distribution. Successively, focusing on the generalization of negation proposed by  Batyrshin \textit{et al.} in  \citep{generating} through the use of specific %decreasing 
functions called  ``negators", \cite{klein2022some} studied several 
properties of these transformations. First, Klein resolved an open problem 
posed in \citep{batyrshin2021contracting}
by proving the structural equivalence between independent negators and linear negators, two specific classes of 
probability distribution negation, introduced
in \citep{generating}, that generalize Yager's negation. Secondly, Klein demonstrated that the repeated application of independent negators
%—--a class to which Yager's negation belongs--—
leads to a monotonic loss of information, as measured by
a general class of entropies, called  $\phi$-entropies in \citep{klein2022some}. 
Furthermore, Klein proved the interesting result that
among the class of independent negators,  Yager’s negation minimizes any $\phi$-entropy. Finally, for some specific classes of $\phi$-entropies, he provided explicit quantitative bounds for this increase in uncertainty.

Parallel to these investigations, 
several researchers have proposed variants of traditional information measures to  evaluate
the properties of probability distributions' negations. 
This line of research has been investigated in \citep{chaudhary2025extension, Deng_extropy2024}
 where
 alternative informational measures, including extropy and varextropy, have been employed.

More recently, several researchers have tackled the problem of
 quantifying the statistical distance between a distribution and its negation. 
%While previous works mainly analyzed the entropy of a negated distribution, 
\cite{tinaztepe2023application} investigated the Kullback-Leibler divergence $D(\bp\|\overline{\bp})$ to measure how much a distribution $\bp$ differs from its Yager's negation $\overline{\bp}$. By extending the Ky Fan inequality \citep{beckenbach2012inequalities}, Tinaztepe established  upper and lower bounds on  the Kullback-Leibler divergence  $D(\bp\|\overline{\bp})$ in terms of Gini entropy. His analysis revealed that the divergence is minimized when the distribution is uniform---where the distribution and its negation become identical---and increases as the probability mass concentrates on specific outcomes.

%%%%%%%%%%%%%%%%%%%%%%%%%%%%%
\subsection{Our Contributions}
In this paper, we provide a comprehensive and in-depth information-theoretical 
analysis of the properties of probability distribution negations, with a focus on Yager's negation. Using
 tools from Information Theory \citep{mceliece2002theory} and Majorization Theory \citep{marshall1979inequalities},
we recover several results from the literature in a simpler and unified way; more importantly,
we provide novel and improved quantitative results with respect to the known literature on the topic.
Our main contributions can be summarized as follows.

In Section \ref{sec:KL_div}, extending the results of \cite{tinaztepe2023application}, we prove that among the class 
of independent negators \citep{generating} (a class of probability distribution negation that subsumes Yager's negator), 
Yager's negator is the one that maximizes the Kullback-Leibler divergence between the original distribution and its negation. 
This result is intuitively
pleasing and reinforces the idea that Yager's negator represents the \textit{canonical} negator among independent negators. 
In fact, since the Kullback-Leibler divergence is the most widely used measure of distance between probability distributions, the fact that Yager’s negator is provably the farthest from the original distribution aligns perfectly with the intuitive notion that a negation should represent the \lq\lq opposite" of the original statement.

In Section \ref{sec:min_entropy}, we recover and generalize the results of \cite{klein2022some} to the entire class of Schur-concave functions \citep{marshall1979inequalities}. More explicitly,
we prove that among the class of independent negators, Yager's negator is the one that minimizes any Schur-concave function. Schur-concave functions constitute an important class of functions
and appear in many problems of pure and applied mathematics.
Furthermore, by leveraging Majorization Theory, we also obtain explicit quantitative bounds that allow for a comparison between different negators.

In Section~\ref{sec:maximal_mutual_information}, we analyze the negation process through the lens of a communication channel. We demonstrate that, within the class of  channels induced by independent negators, the channel associated with Yager's negator maximizes the mutual information. In doing so, it proves to be the operator that preserves the greatest amount of information about the initial distribution.

In Section \ref{sec:increasing}, by leveraging Majorization Theory again \citep{marshall1979inequalities}, we first demonstrate that the application of any independent negator increases uncertainty across the entire class of Schur-concave functions. Subsequently, focusing on the narrower class of $\phi$-entropies, we provide explicit quantitative estimates for the entropy increase resulting from successive applications of an arbitrary independent negator. Finally, we leverage these bounds to explicitly quantify the rate at which iterated applications of an independent negator converge to the uniform distribution.

Taken all together, our results provide evidence that among a wide class of possible ways to define 
the negation of a probability distribution, Yager's proposal seems to be the most sensible one under many criteria.
%%%%%%%%%%%%%%%%%%%%%%%%%%%

\section{Mathematical Preliminaries}\label{sec:pre}

Let 
\begin{align*}
        \cP_n = \{\bp=&(p_1,\dots,p_n)\ :\   \sum_{i=1}^n p_i = 1,\\&\quad\, 0\leq p_i\leq1,\,\forall i=1,\dots,n\}
\end{align*}
be the $(n-1)$-dimensional probability simplex.

{In the most general setting \citep{generating}, 
%a \cite{generating} considered a wider definition of negations by treating them as point-by-point transformations. Formally, 
a negation of a probability distribution can be defined through  a function $\N:\cP_n \to \cP_n$ such that for any distribution $\bp = (p_1,\ldots,p_n) \in \cP_n$, the resulting 
 negated probability distribution 
\begin{equation}
    \N(\bp)=(q_1,\ldots,q_n),
\end{equation}
has the property that  
\begin{equation}\label{eq:property_baty}
    \text{if } p_i\leq p_j, \text{ then } q_i \geq q_j,  \mbox{ for all } 
    i,j\in \{1,\dots,n\}.
\end{equation}

The property  \eqref{eq:property_baty} reflects the intuitive requirement that negation should be order-reversing. Specifically, it suggests that the less likely a particular outcome is in the original distribution, the more prominent its likelihood becomes upon negation, and vice versa.

Each negation $\N$ is characterized by an associated \textit{negator} $N:[0,1]\to[0,1]$, which is a non-increasing function that maps the individual probabilities $p_i$ into the negated probability $q_i$, i.e., such that
$$\N(\bp)=(N(p_1), \ldots, N(p_n)).$$
 Moreover, from the definition of $\N$, it follows that a negator $N$ must satisfy the following properties:
\begin{align}
    &0 \leq N(p_i) \leq 1, \forall i=1,\dots,n,\\
    &\sum_{i=1}^n N(p_i)=1,\\
    &\text{if } p_i\leq p_j, \text{ then } N(p_i) \geq N(p_j). \label{mon}
\end{align}
Under the above constraints, multiple valid negators exist. Two commonly occurring negators are Yager's negator $N_Y$ (described previously)
\begin{equation}\label{defY}
    N_Y(p_i) = \frac{1-p_i}{n-1}, \ \forall i=1,\dots,n,
\end{equation}
and the uniform negator $N_U$ 
\begin{equation}\label{defU}
    N_U(p_i) = \frac{1}{n}, \forall i=1,\dots,n.
\end{equation}
These negators belong to the class of independent negators \citep{generating}.
A negator $N$ is \textit{independent} when the negation of one outcome’s probability does not depend on the probabilities of the other outcomes, i.e., when $N(p_i)$ depends only on the probability $p_i$. Moreover, a negator $N_\alpha$ is \textit{linear} if it can be expressed as a convex combination of the uniform negator $N_U$ and Yager's negator $N_Y$, i.e., if
$$N_\alpha(p_i)=\alpha N_U(p_i) + (1-\alpha)N_Y(p_i),$$ 
for some $\alpha\in[0,1]$.

In a recent work, \cite{klein2022some} established the equivalence between  independent negators and linear negators. More precisely, \textit{any} independent negator $N$ can be expressed as follows:
\begin{equation}
    \label{eq:convex_combination}
    N(p_i) = \alpha\frac{1}{n} + (1-\alpha)\frac{1-p_i}{n-1}, \quad \forall i=1,\dots,n,
\end{equation}
where $\alpha\in[0,1]$.

In the following, we use $N_{\alpha}$ with $\alpha\in[0,1]$ to denote an arbitrary independent negator defined as shown in (\ref{eq:convex_combination}). We also observe that $N_1 = N_U$ and $N_0 = N_Y$. Similarly, we denote by $\N_{\alpha}(\bp)=
(N_\alpha(p_1), \ldots, N_\alpha(p_n))$ the probability distribution obtained through the application of the negator $N_{\alpha}$ to the probability distribution $\bp=(p_1, \ldots , p_n)$.}

We now recall some intermediate results and definitions that we will need later on.

\begin{definition}[\cite{marshall1979inequalities}]
    Given two vectors $\x, \y \in \mathbb{R}_+^n $, we say that $\x \preceq \y$, or equivalently that $\x$ is \textit{majorized} by $\y$  if, and only if, by denoting with $\x^{\downarrow}$ and $\y^{\downarrow}$ the vector ordered in a non-increasing fashion, it holds that
\begin{align*}
    \sum_{i=1}^k x^{\downarrow}_i &\leq \sum_{i=1}^k y^{\downarrow}_i,\:\forall k=1,\dots,n-1,\\&\text{and}\\
    \sum_{i=1}^n  x^{\downarrow}_i &=\sum_{i=1}^n y^{\downarrow}_i.
\end{align*}
\end{definition}

\begin{definition}[\cite{marshall1979inequalities}]\label{def:schur}
We say that a real-valued function $\psi:\cP_n \to \mathbb{R}_+$ is Schur-convex (resp., Schur-concave) if $\psi$ is order preserving (resp., inverse-order preserving) with respect to the partial order $\preceq$, that is, 
\begin{equation*}
    \bp\preceq \bq \Rightarrow \psi(\bp) \leq \psi(\bq)\:\:(\text{resp., }\bp\preceq \bq \Rightarrow \psi(\bp) \geq \psi(\bq))
\end{equation*}
\end{definition}

\begin{definition}\label{def:phi_entropy}
    Let $\phi:[0,1]\to\mathbb{R}$ be a strictly concave function. The $\phi$-entropy $H_\phi$ of a probability distribution $\bp=(p_1,\dots,p_n)\in\cP_n$ is defined as
    \begin{equation}
        H_\phi(p) = \sum_{i=1}^n \phi(p_i).
    \end{equation}
\end{definition}
We note that the class of Schur-concave functions is strictly broader than that of $\phi$-entropies. Indeed, while every $\phi$-entropy is also a Schur-concave function, the converse is not true, as shown in the example below.
\begin{example}
   Some examples of Schur-concave functions are:
    \begin{itemize}
        \item Gini entropy: $H_G(\bp) = 1-\sum_{i=1}^n p^2_i=\sum_{i=1}^n (p_i-p_i^2)$. This is also a $\phi$-entropy generated by $\phi(x)=x-x^2$.
        \item Shannon entropy: $H(\bp) = -\sum_{i=1}^n p_i\log p_i$. This is also a $\phi$-entropy generated by $\phi(x)=-x\log x$. 
        \item Tsallis entropy: $H_q(\bp) = \frac{1}{q-1}(1-\sum_{i=1}^n p^q_i)=\sum_{i=1}^n\frac{p_i-p_i^q}{q-1}$. This is a $\phi$-entropy generated by $\phi(x)=\frac{x-x^q}{q-1}$.
        \item R\'enyi entropy of order $\alpha\in(0,\infty)\setminus\{1\}$: $H_\alpha(\bp) = \frac{1}{1-\alpha}\log\sum_{i=1}^n p^\alpha_i$. This is \textbf{not} a $\phi$-entropy because the outer logarithm prevents it from being expressed in the form $\sum_{i=1}^n \phi(p_i)$ for some appropriate function $\phi$.
    \end{itemize}
\end{example}

\begin{definition}[\cite{marshall1979inequalities}]
A $n\times n$ matrix $A\in\mathbb{R}^{n\times n}$ is said to be \textit{doubly stochastic} if 
\begin{equation}
    {A}_{ij} \geq 0,\quad\forall i,j=1,\dots,n,
\end{equation}
and 
\begin{align}
        \sum_{i=1}^n {A}_{ij} &= 1, \quad\forall j=1,\dots,n,\\
        \sum_{j=1}^n {A}_{ij} &= 1,\quad\forall i=1,\dots,n.
\end{align}

\end{definition}
%%% TODO: aggiungere una frase che introduca brevementi i risultati che sono richiamati su maggiorizzazione, Schur concave functions e così via

\begin{proposition}[{\cite{marshall1979inequalities}[Thm. B.6]}]\label{prop:equivalence_maj}
    Given a pair of vectors $\x,\y\in \mathbb{R}^n$, we have that $\x\preceq\y$ if, and only if, there exists a doubly stochastic matrix $T\in\mathbb{R}_+^{n\times n}$ for which $\x=\y T$.
\end{proposition}

\begin{proposition}[{\cite{marshall1979inequalities}[Prop. A.7.e]}]
\label{Proposition.1}
   Suppose $\y\in \mathbb{R}_{+}^n$ with $\sum_{i=1}^n y_i = 1$ and $\x=\y A$ for some doubly stochastic matrix $A$. Then for all convex functions $\psi:\mathbb{R} \to \mathbb{R}$, it holds that
   \begin{align*}
       \sum_{i=1}^n \psi(x_i) &\leq \lambda(A)n\psi\left(\frac{1}{n}\right) + (1-\lambda(A))\sum_{i=1}^n \psi(y_i)
   \end{align*}
   where
   \begin{equation*}
       \begin{aligned}
           \lambda(A) &= \min_{j,k} \sum_{i=1}^n \min({A}_{ji}, {A}_{ki})\\
           &= 1 - \frac{1}{2}\max_{j, k}\sum_{i=1}^n |{A}_{ji}-{A}_{ki}|
       \end{aligned}
   \end{equation*}
\end{proposition}

\begin{proposition}[{\cite{ho2010interplay}[Thm. 3]}]\label{prop:divergence}
Let $\bp,\bq\in \cP_n$. If $\bq \preceq \bp$, then
\begin{equation}
    H(\bq) - H(\bp) \geq D(\bp \Vert \bq),
\end{equation}
where 
\begin{equation}\label{divdef}
D(\bp \| \bq) = \sum_{i=1}^n p_i \log (p_i/q_i)
\end{equation}
is the Kullback-Leibler divergence from $\bp$ to $\bq$.
\end{proposition}
Throughout this paper, we use $\log$ to denote the logarithm in base 2.

\section{ Yager's negation maximizes the  Kullback-Leibler divergence}\label{sec:KL_div}

In this section, we show that Yager's negator achieves the maximum Kullback-Leibler divergence $D(\bp\|\cdot)$,
where the maximum is computed over all independent 
(equivalently,
linear) negators $N_\alpha$ of $\bp$.
%under the standard element-wise assumption ($q_i=N_{\alpha}(p_i)$).
Formally, for any $\bp\in \cP_n$ it holds that
\begin{equation}    
    \label{eq:max_relative_entropy}
    \N_Y(\bp) = \argmax_{\substack{\N_\alpha, \ 0\leq \alpha\leq 1}} D(\bp\Vert \N_\alpha(\bp)).
\end{equation}

We  prove \eqref{eq:max_relative_entropy} in the following theorem.

\begin{theorem}\label{th:1}
    Let $\bp=(p_1,\dots,p_n)\in \cP_n$ be a probability distribution. Then, among all independent negators, Yager's negator $N_Y$ is  the one that achieves the maximum Kullback-Leibler divergence $D(\bp\Vert \N_Y(\bp))$.
\end{theorem}
\begin{proof}
    To prove our result, let us first 
    demonstrate that, for any fixed $\bp\in\cP_n$, the  {function 
    $$f(\alpha)=D(\bp\Vert \N_\alpha(\bp))$$
    is a convex function on the interval $[0,1]$.}
    {
        Recall from \eqref{eq:convex_combination} that $\N_\alpha(\bp)$ can be expressed as a convex combination of $\N_U(\bp)$ and $\N_Y(\bp)$, that is, 

        \begin{align}
            \N_\alpha(\bp) &= \alpha \N_U(\bp) + (1-\alpha) \N_Y(\bp)\nonumber\\
            &= \alpha(\N_U(\bp) -\N_Y(\bp))+\N_Y(\bp).\label{eq:N_alpha_thm4}
        \end{align}
        From \eqref{eq:N_alpha_thm4}, since $\bp$ is fixed, the mapping $\alpha\mapsto\N_\alpha(\bp)$ is affine in $\alpha$, i.e., 
        it is a linear transformation of $\alpha$ followed by a translation. Furthermore, the Kullback-Leibler divergence $D(\bp\Vert\bq)$ is a convex function with respect to its second argument $\bq$ \citep{CT}. Consequently, because the composition of a convex function with an affine mapping preserves convexity \citep{boyd2004convex}, it follows that $f(\alpha)=D(\bp\Vert \N_\alpha(\bp))$ is convex in $\alpha\in[0,1]$.

        Now, let us show that $f(\alpha)$ is non-increasing in $\alpha\in[0,1]$. For this purpose, we have to prove that its first derivative $f'(\alpha)$ is smaller than or equal to 0 for $\alpha\in[0,1]$. Due to the convexity of $f(\alpha)$, the derivative $f'(\alpha)$ is non-decreasing in $\alpha$, that is,
        \begin{equation}\label{eq:f_convexity}
            f'(\alpha)\leq f'(1), \quad\forall\alpha\in[0,1].
        \end{equation}
        
        Therefore, to establish that $f'(\alpha)\leq 0$ for $\alpha\in[0,1]$, we only need to show that $f'(1)\leq 0$. Let us explicitly compute the derivative $f'(\alpha)$. From (\ref{divdef}), we can rewrite $f(\alpha)=D(\bp\Vert \N_\alpha(\bp))$ as follows:
        \begin{align}
       f(\alpha) =& -\sum_{i=1}^n p_i \log N_\alpha(p_i) +
        \sum_{i=1}p_i\log p_i\nonumber\\=& -\sum_{i=1}^n p_i \log N_\alpha(p_i) -H(\bp).\nonumber
    \end{align}
    Since the term $H(\bp)$ is constant with respect to $\alpha$, and recalling from \eqref{eq:convex_combination} that $N_\alpha(p_i)=(1-\alpha)N_Y(p_i)+\alpha N_U(p_i)$ for each $i=1,\dots,n$, we have that 
    \begin{equation}\label{eq:f_derivative}
        f'(\alpha)=-\frac{1}{\ln 2}\sum_{i=1}^n p_i\frac{\frac{1}{n}-\frac{1-p_i}{n-1}}{N_\alpha(p_i)}.
    \end{equation}
    Since $N_1(p_i)=N_U(p_i)=1/n$ for all $i=1,\dots,n$, evaluating \eqref{eq:f_derivative} at $\alpha=1$ gives
    \begin{align}
        f'(1)&=-\frac{1}{\ln 2}\sum_{i=1}^n p_i\frac{\frac{1}{n}-\frac{1-p_i}{n-1}}{\frac{1}{n}}\nonumber\\
        &=-\frac{n}{\ln 2}\sum_{i=1}^n p_i\left(\frac{1}{n}-\frac{1-p_i}{n-1}\right)\nonumber\\
        &=-\frac{n}{\ln 2}\left(\frac{1}{n}\sum_{i=1}^n p_i - \frac{1}{n-1}\left(\sum_{i=1}^n p_i -p_i^2\right)\right).\label{eq:step_1_f'1}
    \end{align}
    Since $\sum_{i=1}^n p_i=1$, we simplify \eqref{eq:step_1_f'1}:
    \begin{align}
         f'(1)&=-\frac{n}{\ln 2}\left(\frac{1}{n} - \frac{1-\sum_{i=1}^n p_i^2}{n-1}\right)\nonumber\\
         &=-\frac{n}{\ln 2}\left(\frac{n-1-n(1-\sum_{i=1}^n p_i^2)}{n(n-1)}\right)\nonumber\\
         &=-\frac{1}{\ln 2}\left(\frac{n\sum_{i=1}^n p_i^2-1}{n-1}\right)\nonumber\\
         &=-\frac{1}{(n-1)\ln 2}\left(n\sum_{i=1}^n p_i^2-1\right).\label{eq:step_2_f'1}
    \end{align}
    Recall that the Pearson $\chi^2$-divergence between two distributions $\bp$ and $\bq$ is defined as $\chi^2(\bp\Vert\bq)=\sum_{i=1}^n \frac{p_i^2}{q_i}-1$ \citep{sason2016f}. Consequently, when $\bq$ is equal to the uniform distribution $\bu_n=(1/n,\dots,1/n)$, the definition simplifies to $\chi^2(\bp\Vert\bu_n)=n\sum_{i=1}^n p_i^2-1$. This allows us to rewrite the expression \eqref{eq:step_2_f'1} as 
    \begin{equation}
        f'(1)=-\frac{1}{(n-1)\ln 2}\chi^2(\bp\Vert\bu_n).\label{eq:final_step_f'1}
    \end{equation}
    Since $\chi^2(\bp\Vert\bu)\geq 0$, with equality if and only if $\bp=\bu_n$ \citep{sason2016f}, and, for $n\geq 2$, the factor $-\frac{1}{(n-1)\ln 2}$ is strictly negative, it follows from \eqref{eq:f_convexity} and \eqref{eq:final_step_f'1} that
    \begin{equation*}
        f'(\alpha)\leq f'(1)\leq 0, \quad\forall\alpha\in[0,1].
    \end{equation*}
    Consequently, the function $f(\alpha)=D(\bp\Vert\N_\alpha(\bp))$ is non-increasing on $\alpha\in[0,1]$, and the maximum is attained at $\alpha=0$. Thus, given that $\N_0(\bp)=\N_Y(\bp)$, we have shown that $\N_Y(\bp)$ achieves the maximum in
    \begin{equation*}
    \max_{\substack{\N_\alpha, \ 0\leq \alpha\leq 1}} D(\bp\Vert \N_\alpha(\bp)).
    \end{equation*}
    }
\end{proof}

We observe that the property established in \eqref{eq:max_relative_entropy} provides a compelling information-theoretic justification for regarding Yager's negator as the canonical independent negator.
Semantically, a negation should differ \lq\lq as much as possible" from the original. Since the Kullback-Leibler divergence measures how distinguishable one probability distribution is from another, the fact 
that $D(\bp\|\N_Y(\bp))\geq D(\bp\|\N_\alpha(\bp))$, for any $\alpha\in[0,1]$, implies that Yager's negator yields the \lq\lq farthest"
distribution from the input $\bp$. In other words,
Yager's negator creates the strongest semantic possible contrast between a distribution and its negation,
 among all independent negators.

\section{Yager's negator minimizes Schur-concave functions}\label{sec:min_entropy}

In this section, we build upon and expand the findings of \cite{klein2022some}. Specifically,  Klein proved the interesting result
that Yager's negator minimizes \textit{any}  $\phi$-entropy, in the class of all independent (equivalently,
linear)
negators (whose expression is in (\ref{eq:convex_combination})).
We recall that the $\phi$-entropy $H_\phi(\bp)$ of a probability distribution
$\bp=(p_1, \ldots , p_n)$ is defined as $H_\phi(\bp)=\sum_{i=1}^n\phi(p_i)$, where
$\phi:[0,1]\mapsto [0, \infty]$
is strictly concave on $[0, 1]$ (see Definition \ref{def:phi_entropy}).

In this section, we establish a much more general result: Yager's negator minimizes \textit{any} Schur-concave function, a much broader class of function than the $\phi$-entropies. %Furthermore, by leveraging Majorization Theory, we also derive explicit quantitative estimates of the distance between distinct negators.

The following theorem establishes the crucial majorization relationship between the probability distributions generated by two distinct independent negators.

\begin{theorem}\label{th:majorization}
Let $\bp\in\cP_n$ and let $N_\alpha$ and $N_{\alpha'}$, with $\alpha, \alpha'\in[0,1]$, with 
$\alpha' <\alpha$, be two arbitrary linear negators. Then, $\N_\alpha(\bp)$ is majorized by $\N_{\alpha'}(\bp)$, i.e.,
\begin{equation}\label{majalpha}
\N_\alpha(\bp)\preceq \N_{\alpha'}(\bp).
\end{equation}
\end{theorem}
\begin{proof}
Assume without loss of generality that $\bp$ is ordered in a non-increasing fashion, that is, $p_1\geq\dots\geq p_n$. Let $\alpha'=\alpha-\epsilon$.
The two probability distributions that we obtain from the application of the two negators
 $N_\alpha$ and $N_{\alpha'}$, are the following:

\begin{align*}
     \N_\alpha(\bp) = \Biggl(&\alpha\frac{1}{n} + (1-\alpha)\frac{1-p_1}{n-1},\dots,\\&\ \alpha\frac{1}{n} + (1-\alpha)\frac{1-p_n}{n-1}\Biggr)
\end{align*}
and
\begin{align*}
    \N_{\alpha'}(\bp) = \Biggl(&(\alpha-\epsilon)\frac{1}{n} + (1-\alpha+\epsilon)\frac{1-p_1}{n-1},\dots,\\&(\alpha-\epsilon)\frac{1}{n} + (1-\alpha + \epsilon)\frac{1-p_n}{n-1}\Biggr)
\end{align*}

Since $\bp$ is ordered in a non-increasing fashion, by (\ref{mon}) the probability distributions 
$\N_\alpha(\bp)$ and $\N_{\alpha'}(\bp)$ are ordered in a non-decreasing fashion. Therefore, in order to prove that $\N_\alpha(\bp)\preceq \N_{\alpha'}(\bp)$ we have to show that, for each $k=1,\dots,n$, it holds that
\begin{equation}
    \sum_{i=1}^k N_{\alpha}(p_{n-i+1}) \leq \sum_{i=1}^kN_{\alpha'}(p_{n-i+1}).\label{eq:to_prove}
\end{equation}
For an arbitrary $k\in\{1,\dots,n\}$, it holds that
\begin{align}
    \sum_{i=1}^kN_{\alpha'}(p_{n-i+1}) &= \sum_{i=1}^k \frac{\alpha - \epsilon}{n} \nonumber\\&\quad+ (1-\alpha+\epsilon)\frac{1-p_{n-i+1}}{n-1}\nonumber\\
    &=\sum_{i=1}^k \frac{\alpha}{n} + (1-\alpha)\frac{1-p_{n-i+1}}{n-1}\nonumber\\&\quad+ \epsilon\left(\frac{1-p_{n-i+1}}{n-1} - \frac{1}{n}\right)\nonumber\\
    &=\sum_{i=1}^k \frac{\alpha}{n} + (1-\alpha)\frac{1-p_{n-i+1}}{n-1}\nonumber\\&\quad+ \sum_{i=1}^k \epsilon\left(\frac{1-p_{n-i+1}}{n-1} - \frac{1}{n}\right)\nonumber\\
    &\geq \sum_{i=1}^k \frac{\alpha}{n} + (1-\alpha)\frac{1-p_{n-i+1}}{n-1}\label{in}\\
    &=\sum_{i=1}^k N_{\alpha}(p_{n-i+1}),\label{eq:arbitrary_k}
\end{align}
where the inequality (\ref{in}) follows from the fact that 
$$\sum_{i=1}^k \epsilon\left(\frac{1-p_{n-i+1}}{n-1} - \frac{1}{n}\right)\geq 0$$ since 
$$\left(\frac{1}{n},\dots,\frac{1}{n}\right) \preceq \left(\frac{1-p_n}{n-1},\dots,\frac{1-p_1}{n-1}\right).$$

Thereby, since (\ref{eq:arbitrary_k}) holds for any $k\in\{1,\dots,n\}$, we have that $\N_\alpha(\bp)\preceq\N_{\alpha'}(\bp)$.
\end{proof}

As an immediate consequence of the previous theorem, we have that the 
probability distribution generated by Yager's negator \textit{majorizes} 
the probability distribution produced by any other independent negator, as shown in the following corollary.

\begin{corollary}\label{cor:majorization}
For any independent negator $N_\alpha$ with $\alpha\in[0,1]$, and any $\bp\in\cP_n$, it holds that $\N_\alpha(\bp) \preceq \N_Y(\bp)$.
\end{corollary}
\begin{proof}
The proof follows from (\ref{majalpha}) of Theorem \ref{th:majorization}, which  implies that for any 
$\alpha\in[0,1]$ it holds that
\begin{equation}\label{n1}
\N_\alpha(\bp) \preceq \N_0(\bp)=\N_Y(\bp)
\end{equation}
%Therefore, the negation associated with the negator having the smallest possible $\alpha$, i.e., $\alpha=0$, majorizes all the others. 
\end{proof}

Leveraging Corollary \ref{cor:majorization} and the fact that Schur-concave functions are order-reversing with respect to majorization (Definition \ref{def:schur}), we can conclude that Yager's negator minimizes any Schur-concave function among all independent negators, as established in the following corollary.

\begin{corollary}
\label{cor:phi_function}
Let $\bp\in\cP_n$.
For any independent negator $N_\alpha$ with $\alpha\in[0,1]$,  and for any Schur-concave function $\psi$, it holds that $\psi(\N_\alpha(\bp)) \geq \psi(\N_Y(\bp))$.
\end{corollary}
\begin{proof}
It follows from Corollary \ref{cor:majorization} and Definition \ref{def:schur}.
\end{proof}

Since the Shannon entropy is a Schur-concave function, from  Corollary \ref{cor:phi_function} we get that for any arbitrary probability distribution $\bp\in\cP_n$, it holds that
\begin{equation}
  \forall \alpha\in[0,1] \quad  H(\N_\alpha(\bp)) \geq H(\N_Y(\bp)).
\end{equation}
Because Shannon entropy quantifies the overall uncertainty of a distribution, minimizing it implies that Yager's negator produces the most ``structured" possible
%and ``informative" 
negation among the entire class of independent negators. Thus, while the mathematical act of negation inherently redistributes probability mass and increases overall uncertainty (a property we analyze in greater detail
in Section \ref{sec:increasing}), Yager's approach ensures this unavoidable loss of information is minimized. Consequently, it preserves the maximum amount of information about the original distribution.

Furthermore, we observe that by leveraging the majorization relationship established in Theorem \ref{th:majorization}, the result obtained can be strengthened for specific functions. For instance, in the case of Shannon entropy, combining Theorem \ref{th:majorization} with Proposition \ref{prop:divergence} yields that, for any non-uniform probability distribution $\bp\in\cP_n$ and any pair of independent negators $N_\alpha$ and $N_{\alpha'}$ such that $0 \le \alpha' < \alpha \le 1$, the following holds:
\begin{align*}
    H(\N_\alpha(\bp)) &\ge H(\N_{\alpha'}(\bp)) + D(\N_{\alpha'}(\bp) \| \N_\alpha(\bp))\\ &> H(\N_{\alpha'}(\bp)),
\end{align*}
where the strict inequality holds since $D(\N_{\alpha'}(\bp) \| \N_\alpha(\bp))=0$ if and only if 
it were $\N_{\alpha'}(\bp)=\N_{\alpha}(\bp)$, which is not the case since $\alpha'\neq \alpha$.
%In Section \ref{sec:increasing},  Majorization Theory will allow  us to extend and improve Corollary \ref{cor:phi_function}, in the case in which  the function $\psi$ is the Shannon entropy.
% In fact,  from Proposition \ref{prop:equivalence_maj}, and Proposition \ref{Proposition.1}, we can strengthen the result of Corollary \ref{cor:phi_function}.

\section{Yager's negator maximizes the mutual  information}\label{sec:maximal_mutual_information}

As proved in the previous sections, Yager's negator exhibits several interesting properties within the class of independent negators. Specifically, Yager's negator maximizes the distance from the initial distribution in terms of Kullback-Leibler divergence (Section \ref{sec:KL_div}) and minimizes any Schur-concave function (Section \ref{sec:min_entropy}). This section presents another significant property of Yager's approach.

%In general, any negation process transforms a probability distribution into its negation.
In this section, we model the  transformation of a probability distribution into
its negation by means of 
a discrete memoryless communication channel \cite[Ch. 7]{CT}.
Under this formulation, every independent negator $N_\alpha$ yields an associated channel. We demonstrate that the channel associated with Yager's negator $N_Y$ is the one that achieves the maximum mutual information
between the channel input and the output. 

Let us start by formalizing our setting. Let $X$ and $Y$ be two discrete random variables taking values in the sample spaces $\mathcal{X}=\{x_1,\dots,x_n\}$ and $\mathcal{Y}=\{y_1,\dots,y_n\}$, respectively. The random variable $X$ represents the input of the channel, and it is distributed according to a  distribution $\bp=(p_1,\dots,p_n)\in\cP_n$; that is, $P(X=x_i)=p_i$ for each $i=1,\dots,n$. Conversely, the random variable $Y$ represents the output of the channel. The behavior of the channel is characterized by the transition probabilities ${P(Y=y_j\mid X=x_i)}$, which represent the conditional probability of observing the negated outcome $y_j$ given that the input is $x_i$. Because the channel acts as a negation operator, and any independent negator can be expressed as the convex combination of the uniform negator and Yager's negator (as shown in \cite{klein2022some}), it follows that for each independent negator $N_\alpha$ with $\alpha\in[0,1]$, the transition probabilities of its associated channel are defined as follows:
\begin{equation}\label{eq:transition_prob}
    P(Y=y_j\mid X=x_i) = \begin{cases}
        \frac{\alpha}{n} \quad&\mbox{if } j=i\\
        \frac{\alpha}{n} + \frac{1-\alpha}{n-1}\quad&\mbox{if } j\neq i.
    \end{cases}
\end{equation}
From \eqref{eq:transition_prob}, one can verify that the transition probabilities of the channel are perfectly consistent with the corresponding negation process. In fact, the output distribution of $Y$ corresponds to
\begin{align*}
    P(Y=y_j)&=\sum_{i=1}^n P(Y=y_j\mid X=x_i)P(X=x_i)\\
    &=\alpha\frac{1}{n}p_j \\&\quad+ \sum_{i\neq j}\Bigg(\alpha\frac{1}{n}+ (1-\alpha)\frac{1}{n-1}\Bigg)p_i\\
    &= \sum_{i=1}^n \alpha\frac{1}{n}p_i + \sum_{i\neq j} (1-\alpha)\frac{1}{n-1}p_i\\
    &=\frac{\alpha}{n} \sum_{i=1}^n p_i + \frac{1-\alpha}{n-1}\sum_{i\neq j}p_i\\
    &=\frac{\alpha}{n} + (1-\alpha)\frac{1-p_j}{n-1}\\
    &=N_\alpha(p_j)=N_\alpha(P(X=x_j)),
\end{align*}
which exactly yields the negated distribution $\N_\alpha(\bp)=(N_\alpha(p_1),\dots,N_\alpha(p_n))$.
Note that by substituting $\alpha=0$ and $\alpha=1$ into \eqref{eq:transition_prob}, we also recover the transition probabilities for the channels associated with Yager's negator and the uniform negator, respectively.

Having defined the class of channels that we want to consider, let us recall the concept of mutual information $I(X;Y)$ between the input $X$ and the output $Y$ of the communication channel. The mutual information $I(X;Y)$ quantifies the amount of information the output $Y$ gives about the input $X$, and is defined as follows \citep{CT}:
\begin{align}
    I(X;Y) = &\sum_{i=1}^n \sum_{j=1}^n P(Y=y_j\mid X=x_i)\nonumber\\&\quad P(X=x_i)\log \frac{P(Y=y_j\mid X=x_i)}{P(Y=y_j)}.\label{eq:mutual_information}
\end{align}
In the remainder of the section, we demonstrate that, for any 
distribution $\bp$ of $X$, the transition probabilities of the channel associated with Yager's negator $N_Y$ maximizes the quantity $I(X;Y)$, among all possible independent negators.

To prove our claim, for any $\alpha\in[0,1]$, let us denote with 
$Y_\alpha$ the output of the channel associated with the negator $N_\alpha$,
that is, the communication channel with transition probabilities given in (\ref{eq:transition_prob}). Consequently, $Y_\alpha$ is distributed according to $\N_\alpha(\bp)$ and $I(X;Y_{\alpha})$ corresponds to the mutual information evaluated under the transition probabilities (\ref{eq:transition_prob}).
%$P(Y_\alpha\mid X)$ that characterize the channel associated with the negator $N_\alpha$ (cfr. (\ref{eq:transition_prob})). 
We show the following property of  Yager's negator:% solves the following problem:
\begin{equation}\label{eq:prob_I(X;Y)}
    N_Y = \underset{N_\alpha,\:  0\leq\alpha\leq 1}{\argmax} I(X; Y_\alpha).
\end{equation}

The result \eqref{eq:prob_I(X;Y)} is formally proved in the following theorem.
\begin{theorem}\label{th:max_mutual_information}
    Let $X\in \mathcal{X}$ be a discrete random variable distributed according to a probability distribution $\bp\in\cP_n$. For any $\alpha\in[0,1]$, let $Y_\alpha\in\mathcal{Y}$ denote the output of the channel associated with the independent negator $N_\alpha$, given the input $X$. Then, we have 
    \begin{equation}
        I(X;Y_0)\geq I(X;Y_{\alpha}),
    \end{equation}
    where we recall that $Y_0$ is the output of the channel associated with Yager's negator ($N_Y=N_0$).
\end{theorem}
\begin{proof}
    The proof relies on the property that, given an input distribution $P(X)$, the mutual information $I(X;Y)$ is a convex function of the conditional transition probabilities $P(Y\mid X)$ (see \cite{mceliece2002theory}).

    Let $Y_0$ and $Y_1$ be the outputs of the channels associated with Yager's negator ($\alpha=0$) and the uniform negator ($\alpha=1$), respectively. We recall that their transition probabilities $ P(Y_0=y_j\mid X=x_i)$ and $ P(Y_1=y_j\mid X=x_i)$ are defined as follows:
    \begin{equation*}
        P(Y_0=y_j\mid X=x_i)=\begin{cases}
            0\quad&\mbox{if } i=j\\
            \frac{1}{n-1}\quad&\mbox{if } i\neq j,
        \end{cases}
    \end{equation*}
    and
    \begin{equation*}
         P(Y_1=y_j\mid X=x_i) = \frac{1}{n} \quad\forall i,j.
    \end{equation*}

From   \eqref{eq:transition_prob}, the transition probabilities  ${P(Y_\alpha\mid X)}$ for any independent negator $N_\alpha$ with $\alpha\in[0,1]$ can be expressed as a convex combination of the transition probabilities ${P(Y_0=y_j\mid X=x_i)}$ and ${P(Y_1=y_j\mid X=x_i)}$. That is, for every $i$ and $j$:
    
    \begin{align}\label{eq:convex_prob}
        P(Y_\alpha =y_j\mid X=& x_i)=\nonumber \alpha P(Y_1=y_j\mid X=x_i)\nonumber\\&\quad+ (1-\alpha)P(Y_0=y_j\mid X=x_i).
    \end{align}

    Consequently, from \eqref{eq:convex_prob} and since the mutual information $I(X;Y_\alpha)$ is convex in the transition probabilities $P(Y_\alpha\mid X)$ \cite[Theorem 1.7]{mceliece2002theory}, we obtain
    \begin{equation}\label{eq:upper_bound_mutual}
        I(X;Y_\alpha)\leq \alpha I(X;Y_1) + (1-\alpha)I(X;Y_0).
    \end{equation}
    We observe that for the uniform negator, the transition probabilities $P(Y_1=y_j\mid X=x_i) = 1/n$ are completely independent of the input $X$. Consequently, the output $Y_1$  preserves no information about the input $X$, and \eqref{eq:mutual_information}  yields $I(X;Y_1)=0$. Substituting this result into \eqref{eq:upper_bound_mutual},  we obtain
    \begin{equation*}
        I(X;Y_\alpha)\leq (1-\alpha)I(X;Y_0)\leq I(X;Y_0),
    \end{equation*}
    which concludes the proof.
\end{proof}

The result of Theorem \ref{th:max_mutual_information} offers an interesting interpretation. By maximizing the mutual information, Yager's negator effectively acts as the least ``noisy" channel (or operator) within the class of independent negators. This reveals an interesting duality: although Yager's negator pushes the distribution as far away as possible from the original (as demonstrated in Section \ref{sec:KL_div}), it does so in the most structured and least destructive way. Consequently, among all independent negators, Yager's approach achieves the maximum distance while preserving the highest possible amount of information about the initial distribution.

\section{Quantifying the increase of uncertainty under negation}\label{sec:increasing}

In this section, we address the following question: how much does the \textit{uncertainty} of a probability distribution increase under the application of a negation? Previous studies \citep{klein2022some,Yager2015} have shown that applying an independent negator generally increases uncertainty, as quantified by various entropy measures. In this section,
we take this analysis a step further. First, through a straightforward Majorization Theory argument, we demonstrate that the application of an independent negator generates an increase in uncertainty 
as measured by an entire class of Schur-concave functions. Subsequently, focusing on the narrower class of $\phi$-entropies, we strengthen this result by explicitly quantifying the amount by which the entropy increases.

Moreover, since repeated applications of a negator are known to converge to the uniform distribution \citep{generating,Yager2015},  we also quantify  the rate at which iterative negations converge to the uniform distribution.%for the class of $\phi$-entropies.

Let us start by demonstrating that the value of any Schur-concave function increases under the application of an independent negator. For this purpose, we show that the probability distribution obtained through an independent negator is majorized by its initial, un-negated,  distribution.

\begin{theorem}\label{th:p_maj_N}
    Let $\bp\in \cP_n$ be an arbitrary probability distribution and let $N_{\alpha}$ be an arbitrary independent negator with $\alpha\in[0,1]$. Then,  the negation $\N_\alpha(\bp)$ is majorized by $\bp$, that is,  $\N_{\alpha}(\bp) \preceq \bp$.
\end{theorem}
\begin{proof}
    By Corollary \ref{cor:majorization}, we know that $\N_\alpha(\bp)\preceq \N_Y(\bp)$ for any independent negator. Therefore, due to the transitivity of majorization, it suffices to prove that $\N_{Y}(\bp) \preceq \bp$ to establish that  $\N_{\alpha}(\bp) \preceq \bp$.

    For this purpose, according to Proposition \ref{prop:equivalence_maj}, proving that $\N_Y(\bp) \preceq \bp$ is equivalent to showing that there exists a doubly-stochastic matrix $T\in\mathbb{R}_+^{n\times n}$ such that $\N_Y(\bp)=\bp T$. Let $T$ be defined as follows:
    \begin{equation}
    T_{ij} = 
    \begin{cases}
        0 &\text{if } i=j,\\
        \frac{1}{n-1} &\text{otherwise.}
    \end{cases}
\end{equation}
    By construction, the sum of each row and each column of $T$ is equal to $1$, confirming that $T$ is doubly-stochastic. Furthermore, 
    from formula (\ref{defY}) one can immediately verify that the matrix $T$ satisfies: $\N_Y(\bp) = \bp T$. Consequently, by Proposition \ref{prop:equivalence_maj}, it holds that $\N_Y(\bp)\preceq \bp$, concluding the proof.
\end{proof}

Given two arbitrary probability distributions $\bp,\bq\in\cP_n$, it is well known  \citep{marshall1979inequalities} that the 
relation $\bp\preceq\bq$ is a mathematical rigorous formalization of the 
intuitive concept that $\bp$ is \lq\lq more flat than" $\bq$.
Therefore, the majorization relationship established in Theorem \ref{th:p_maj_N} provides a formalization of the intuitive concept that a negation intrinsically acts as a smoothing operator. By reallocating probability mass from highly likely outcomes to less likely alternatives, any independent negator effectively ``flattens" the original distribution. Since Schur-concave functions are, by definition, order-reversing with respect to majorization (see Definition \ref{def:schur}), the result of Theorem \ref{th:p_maj_N} guarantees that the value of a Schur-concave function can only increase upon negation. This leads directly to the following result.

\begin{corollary}\label{cor:schur_concave_increase}
    Let $\bp\in \cP_n$ be an arbitrary probability distribution and let $N_{\alpha}$ be an arbitrary independent negator with $\alpha\in[0,1]$. Then, for any Schur-concave function $\psi$, it holds that $\psi(\N_\alpha(\bp))\geq \psi(\bp)$.
\end{corollary}
\begin{proof}
    It follows from Theorem \ref{th:p_maj_N} and Definition \ref{def:schur}.
\end{proof}

%% Riscrivere in termini di phi-entropies

Shifting our focus from the broad class of Schur-concave functions to the subclass of $\phi$-entropies, we can strengthen the result of Corollary \ref{cor:schur_concave_increase}. Recall that the $\phi$-entropy $H_\phi(\bp)$ of a probability distribution $\bp=(p_1,\dots,p_n)\in\cP_n$ is defined as $H_\phi(\bp)=\sum_{i=1}^n\phi(p_i)$, where
$\phi:[0,1]\to [0, \infty]$ is a strictly concave function on $[0, 1]$ (see Definition \ref{def:phi_entropy}).

In what follows, we provide a quantitative estimate on the $\phi$-entropy increase resulting from successive applications of an
arbitrary independent negator.

\begin{theorem}\label{th:lb_alpha}
    Let $\bp\in \cP_n$ be an arbitrary probability distribution and $H_\phi$ be an arbitrary $\phi$-entropy. For any $\alpha\in[0,1]$, let $\N_\alpha^{(i)}(\bp)$ be the probability distribution obtained after $i$ applications of the independent negator $N_\alpha$ to $\bp$. Then, it holds that
\begin{align}\label{eq:lb_4}
    H_\phi(\N_\alpha^{(i)}(\bp)) \geq &\left(1-\left(\frac{1-\alpha}{n-1}\right)^i\right)H_\phi(\bu_n) \nonumber \\ &+\left(\frac{1-\alpha}{n-1}\right)^i H_\phi(\bp),
\end{align}
where $\bu_n=(1/n,\dots,1/n)\in\cP_n$ denotes the uniform distribution.
\end{theorem}
\begin{proof}
    {We prove \eqref{eq:lb_4} by induction on $i$ using Proposition \ref{Proposition.1}.}
    We recall that Proposition \ref{Proposition.1} states that for any vector $\y\in\mathbb{R}_+^n$ with $\sum_{\ell=1}^n y_{\ell}=1$ and $\x=\y T$ for some doubly stochastic matrix $T\in\mathbb{R}_+^{n\times n}$, the following inequality holds for any convex function $\psi:\mathbb{R}\to \mathbb{R}$:
    \begin{equation}\label{eq:prop2recall_2}
        \sum_{\ell=1}^n\psi(x_{\ell})\leq \lambda(T)n\psi\left(\frac{1}{n}\right)+(1-\lambda(T))\sum_{\ell=1}^n\psi(y_{\ell}),
    \end{equation}
    where the coefficient $\lambda(T)$ is defined as $\lambda(T)=1-\frac{1}{2}\max_{j,k}\sum_{\ell=1}^n|T_{j\ell}-T_{k\ell}|$.

    \smallskip
    {
    \noindent
    \textbf{Base case} ($i=1$): We prove inequality \eqref{eq:lb_4} for $i=1$ using Proposition \ref{Proposition.1}.
    To apply \eqref{eq:prop2recall_2} of Proposition \ref{Proposition.1}, we  need to
    specify the doubly stochastic matrix $T$ for which  $\N_\alpha^{(1)}(\bp)=\bp T$. For this purpose, we recall that any independent negator $N_\alpha$ can be expressed as
    $N_\alpha = (1-\alpha)N_Y + \alpha N_U$ (see \eqref{eq:convex_combination}). Consequently, the application of the negator $N_\alpha$ can be represented by right-multiplying the distribution
    $\bp$ by the matrix $A_\alpha = (1-\alpha)A_Y + \alpha A_U$, where $A_Y$ is the doubly stochastic matrix associated with Yager's negator (defined as $(A_Y)_{\ell m}=\frac{1}{n-1}$ for $\ell\neq m$ and $0$ otherwise), and $A_U$ is the matrix, where all entries are equal to $1/n$. Thus, we can write the components of $A_\alpha$ as
    \begin{equation}\label{eq:A_alpha_def}
        (A_\alpha)_{k\ell} = \begin{cases}
            \frac{\alpha}{n}&\quad\mbox{if } k=\ell,\\
            \frac{\alpha}{n} + \frac{1-\alpha}{n-1}&\quad\mbox{if } k\neq\ell.
        \end{cases}
    \end{equation}
    We recall from Birkhoff's theorem \citep[Thm. A.2]{marshall1979inequalities} that the set of doubly stochastic matrices is the convex hull of permutation matrices. Therefore, since $A_\alpha$ is the convex combination of doubly stochastic matrices, it follows that $A_\alpha$ is itself a doubly stochastic matrix. Hence, $A_\alpha$ is the doubly stochastic matrix that we need to apply Proposition \ref{Proposition.1}. Let us evaluate the coefficient
    \begin{equation}\label{eq:coeff_def_A}
        \lambda(A_\alpha)=1-\frac{1}{2}\max_{j,k}\sum_{\ell=1}^n|(A_\alpha)_{j\ell}-(A_\alpha)_{k\ell}|.
    \end{equation}
    Specifically, for any two distinct rows $j$ and $k$, we must compute the absolute difference between their entries in each column. 
    From \eqref{eq:A_alpha_def}, for any $j\neq k$, we have that
    \begin{equation}
        |(A_\alpha)_{j\ell}-(A_\alpha)_{k\ell}|=\begin{cases}
            \frac{1-\alpha}{n-1} &\quad\mbox{if } \ell=k \mbox{ or } \ell=j,\\
            0 &\quad\mbox{otherwise}.\nonumber
        \end{cases}
    \end{equation}
    Summing over all column indices $\ell$, for any $j\neq k$, it follows that
    \begin{equation}\label{eq:sum_def}
        \sum_{\ell=1}^n|(A_\alpha)_{j\ell}-(A_\alpha)_{k\ell}| = 2\left(\frac{1-\alpha}{n-1}\right).
    \end{equation}
    Thus, from \eqref{eq:coeff_def_A} and \eqref{eq:sum_def} we have that
    \begin{equation}
        \lambda(A_\alpha)=1-\frac{1}{2}2\left(\frac{1-\alpha}{n-1}\right)=1-\frac{1-\alpha}{n-1}.
    \end{equation}
    We can now apply inequality \eqref{eq:prop2recall_2} of Proposition \ref{Proposition.1} by setting $\y=\bp$, $\x=\N_\alpha^{(1)}(\bp)$, $T=A_\alpha$, and selecting the function $\psi(x)=-\phi(x)$. To justify the choice of $\psi(x)=-\phi(x)$, we note that because $\bp$ and $\N_\alpha^{(1)}(\bp)$ are probability distributions, their components, as well as $1/n$, lie in $[0,1]$. Thus, since $\phi(x)$ is a strictly concave function on $[0,1]$, it follows that $\psi(x)=-\phi(x)$ is a convex function on this same interval. From this substitution, we obtain
    \begin{align*}
    H_\phi(\N_\alpha^{(1)}(\bp)) &\geq \left(1-\frac{1-\alpha}{n-1}\right)H_\phi(\bu_n)\\&\quad+\left(\frac{1-\alpha}{n-1}\right) H_\phi(\bp),
    \end{align*}
    which proves the base case.

    \smallskip
    \noindent
    \textbf{Inductive step}: Assume that inequality \eqref{eq:lb_4} holds for $i$, we show that it holds for $i+1$. Note that $\N_\alpha^{(i+1)}(\bp)=\N_\alpha^{(i)}(\bp)A_\alpha$, where $A_\alpha$ is the doubly stochastic matrix defined in \eqref{eq:A_alpha_def}. Thereby, proceeding analogously to the base case, we can apply  Proposition \ref{Proposition.1} by setting $\y=\N_\alpha^{(i)}(\bp)$, $\x=\N_\alpha^{(i+1)}(\bp)$, $T=A_\alpha$, and selecting the function $\psi(x)=-\phi(x)$, which yields
    \begin{align}\label{eq:ind_step_1}
         H_\phi(\N_\alpha^{(i+1)}(\bp)) &\geq \left(1-\frac{1-\alpha}{n-1}\right)H_\phi(\bu_n)\nonumber\\&\quad+\left(\frac{1-\alpha}{n-1}\right) H_\phi(\N_\alpha^{(i)}(\bp)).
    \end{align}
    From the inductive hypothesis it holds that
    \begin{align}\label{eq:ind_H}
    H_\phi(\N_\alpha^{(i)}(\bp)) \geq &\left(1-\left(\frac{1-\alpha}{n-1}\right)^i\right)H_\phi(\bu_n) \nonumber \\ &+\left(\frac{1-\alpha}{n-1}\right)^i H_\phi(\bp),
    \end{align}
    Substituting \eqref{eq:ind_H} into  \eqref{eq:ind_step_1} gives
    \begin{align*}
         H_\phi(\N_\alpha^{(i+1)}(\bp)) &\geq \left(1-\frac{1-\alpha}{n-1}\right)H_\phi(\bu_n)\\&\!\!+\left(\frac{1-\alpha}{n-1}\right) \left(1-\left(\frac{1-\alpha}{n-1}\right)^i\right)H_\phi(\bu_n)\\
         &\!\!+\left(\frac{1-\alpha}{n-1}\right)^{i+1}H_\phi(\bp)\\
         &=\left(1-\left(\frac{1-\alpha}{n-1}\right)^{i+1}\right)H_\phi(\bu_n) \nonumber \\ &\quad+\left(\frac{1-\alpha}{n-1}\right)^{i+1} H_\phi(\bp),
    \end{align*}
    which completes the inductive step and concludes the proof.
    }
\end{proof}
We note that Theorem \ref{th:lb_alpha} provides a quantitative counterpart to the results established in  \citep{klein2022some}. Specifically, compared to \cite{klein2022some}, the theorem establishes an explicit lower bound for the value of an arbitrary $\phi$-entropy after $i$ successive applications of an independent negator, including Yager's negator ($\alpha=0$). 

Furthermore, the bound established in Theorem \ref{th:lb_alpha} provides a quantitative estimate  of the convergence rate of successive applications of an independent negator to the uniform distribution. {This convergence is illustrated in Figure \ref{fig:uniform_convergence}  for the specific case of Yager's negator}. To demonstrate this, let us analyze the behavior of the lower bound \eqref{eq:lb_4} for the specific case of Shannon entropy $H(\cdot)$ as the number of iterations $i \to \infty$ (though the same reasoning naturally extends to any other $\phi$-entropy). For any $n>2$, and $\alpha\in[0,1]$,  the term $\frac{1-\alpha}{n-1}$ lies in the interval $[0,1)$. Consequently, the term $\left(\frac{1-\alpha}{n-1}\right)^i$ vanishes as the integer $i$ approaches infinity. Therefore, from \eqref{eq:lb_4} we have
\begin{align*}
    \lim_{i \to \infty} H(\N_\alpha^{(i)}(\bp)) \geq&\lim_{i \to \infty} \Biggl(\left(1-\left(\frac{1-\alpha}{n-1}\right)^i\right)H(\bu_n) \\&\quad+\left(\frac{1-\alpha}{n-1}\right)^i H(\bp)\Biggr)\\
    =& (1-0)\log n + 0 \cdot H(\bp) = \log n.
\end{align*}
Since the Shannon entropy of any $n$-dimensional probability distribution is upper bounded by $\log n$, it follows that $\lim_{i \to \infty} H(\N_\alpha^{(i)}(\bp)) = \log n$. Moreover, because the uniform distribution is the unique distribution that maximizes the Shannon entropy, this confirms that the sequence of iterated distributions $\N_\alpha^{(i)}(\bp)$ converges asymptotically to the uniform distribution. Furthermore, the term $\left(\frac{1-\alpha}{n-1}\right)^i$ explicitly quantifies the rate of this convergence. {Figure \ref{fig:thm8} illustrates this rate of convergence comparing the exact value of the Shannon entropy with the corresponding lower bound provided in \eqref{eq:lb_4}.}

\begin{figure}[h]
    \centering
    \includegraphics[width=1\linewidth]{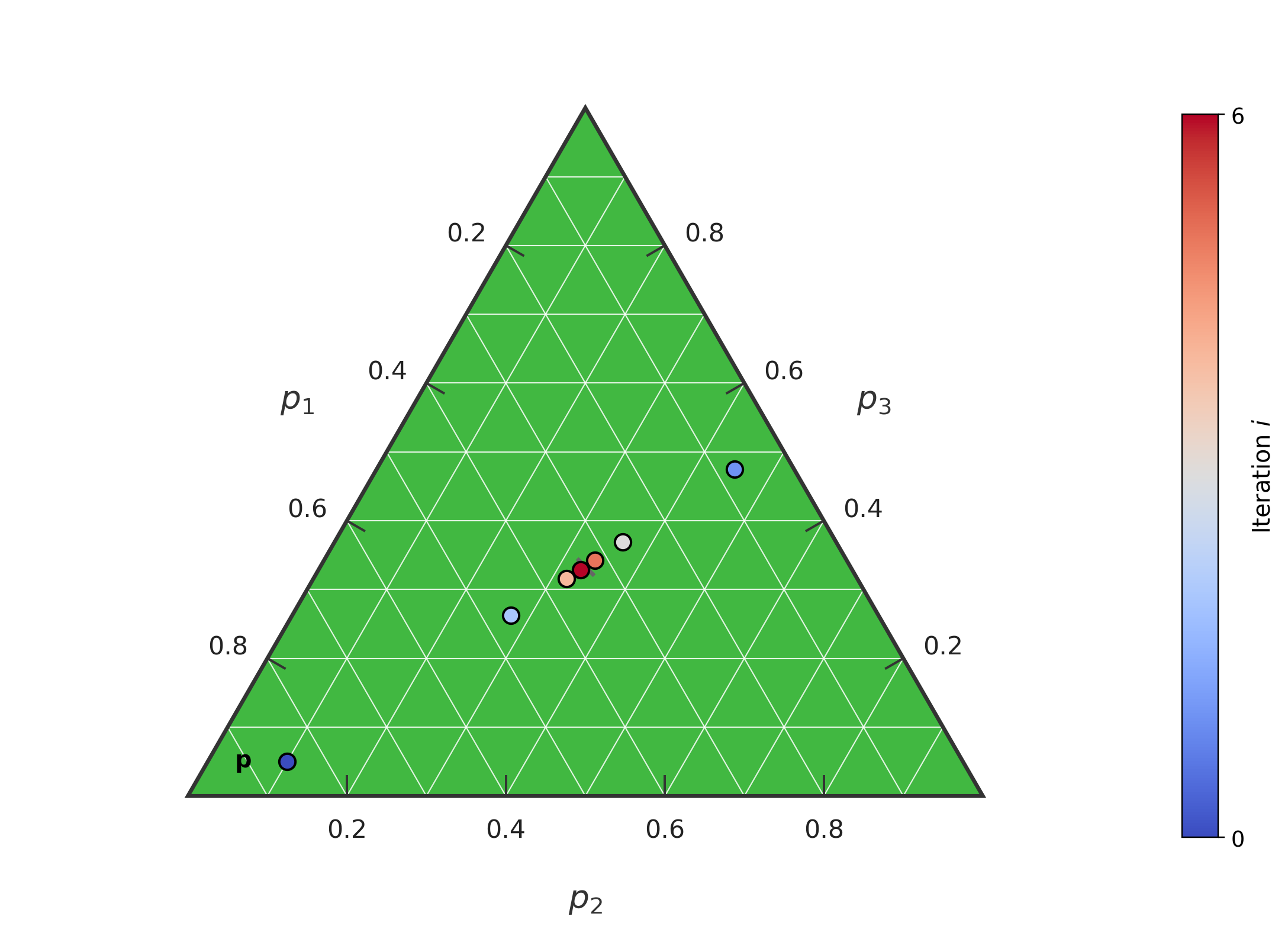}
    \caption{Visualisation of the dynamics of $\N_Y^{(i)}(\bp)$, for 
    $i=0,\dots,6$ and 
    $\bp=(0.85, 0.1, 0.05)$,  in the 2-dimensional simplex.}
    \label{fig:uniform_convergence}
\end{figure}

\begin{figure}[h]
    \centering
    \includegraphics[width=1\linewidth]{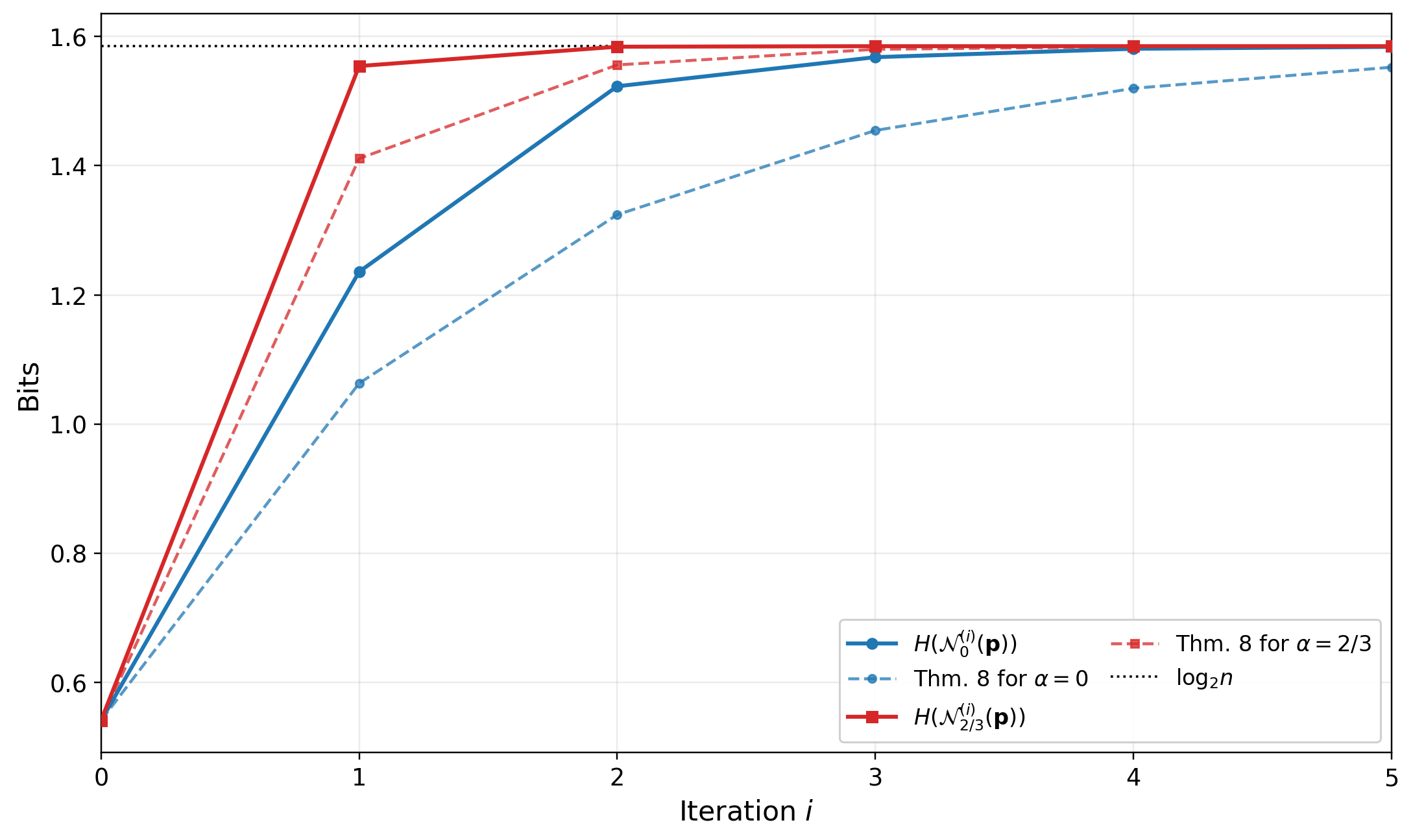}
    \caption{Comparison of the convergence rates to the uniform distribution under successive applications of independent negators ($\alpha=0$ and $\alpha=2/3$), starting from the initial distribution $\bp=(0.9, 0.08, 0.02)$.   The plot compares the exact Shannon entropy values (solid lines) with the corresponding lower bounds established in Theorem \ref{th:lb_alpha} (dashed lines). The $x$-axis denotes the number of iterations $i$ and the $y$-axis the corresponding entropy value.}
    \label{fig:thm8}
\end{figure}

\section{Conclusions}
In this work, we have presented a unified and comprehensive information-theoretic analysis of probability distribution negation, with particular emphasis on Yager’s formulation. By leveraging tools from information theory and majorization theory, we have established a set of strong and mutually reinforcing results that clarify the structural and operational role of negation in probabilistic settings.

Our analysis shows that, within the broad class of independent (equivalently, linear) negators, Yager’s negator enjoys a number of optimality properties. It maximizes the Kullback–Leibler divergence from the original distribution, thereby providing the strongest notion of \lq\lq opposition” in an information-theoretic sense. At the same time, it minimizes all Schur-concave functions, including entropy measures, ensuring that the unavoidable increase in uncertainty induced by negation is as controlled as possible. Furthermore, when interpreted as a communication channel, Yager’s negator maximizes mutual information, revealing a fundamental duality: it produces the most contrasting distribution while preserving the largest amount of information about the original one. Finally, we have quantified the increase in uncertainty under repeated negation and characterized the rate of convergence toward the uniform distribution.

Taken together, these results provide a  multifaceted justification for considering Yager’s negator as the canonical choice among a wide class of admissible negation operators.

Several directions for future research naturally emerge from this work. First, it would be of interest to extend the present analysis beyond independent negators to more general, possibly context-dependent transformations, where the negation of each component depends on the full distribution \citep{batyrshin2021negations, Wu_exp, Zhang}. Second, the interaction between negation and alternative divergence measures (e.g., Wasserstein distances or $f$-divergences) deserves further investigation. Third, exploring the role of probabilistic negation in modern machine learning—particularly in settings such as learning from complementary labels, uncertainty quantification, and robust inference—may lead to new algorithmic insights. Finally, applications to multi-agent systems, information fusion, and decision-making under uncertainty represent promising venues where the theoretical properties established here could be exploited in practice.

\backmatter

%\bmhead{Acknowledgements}

\bibliography{sn-bibliography}

@article{batyrshin2021contracting,
  title={Contracting and involutive negations of probability distributions},
  author={Batyrshin, I.},
  journal={Mathematics},
  volume={9},
  number={19},
  pages={2389},
  year={2021},
  publisher={MDPI},
  doi = {10.3390/math9192389}
}

@article{batyrshin2021negations,
  title={Negations of probability distributions: a survey},
  author={Batyrshin, I. and Kubysheva, N. and Bayrasheva, V. and Kosheleva, O. and Kreinovich, V.},
  journal={Computaci{\'o}n y Sistemas},
  volume={25},
  pages={775--781},
  year={2021}
}

@article{generating,
  author  = {Batyrshin, I. and Villa-Vargas, L. and Ramirez-Salinas, M. and Salinas-Rosales, M. and Kubysheva, N.},
  title   = {Generating negations of probability distributions},
  journal = {Soft Computing},
  year    = {2021},
  volume  = {25},
  pages   = {7929--7935},
  doi     = {10.1007/s00500-021-05802-5},
  publisher = {Springer}
}

@book{beckenbach2012inequalities,
  author    = {Beckenbach, E. F. and Bellman, R.},
  title     = {Inequalities},
  year      = {1961},
  publisher = {Springer Science \& Business Media},
  address   = {Berlin},
  isbn      = {978-3-642-64971-4}
}

@book{boyd2004convex,
  title={Convex optimization},
  author={Boyd, S. and Vandenberghe, L.},
  year={2004},
  publisher={Cambridge university press}
}

@article{chaudhary2025extension,
  author  = {Chaudhary, S. and Sahu, P. and Gupta, N.},
  title   = {Extension of {Yager}'s negation of probability distribution based on uncertainty measures},
  journal = {arXiv preprint arXiv:2504.04762},
  year    = {2025},
  eprint  = {2504.04762},
  archivePrefix = {arXiv},
  primaryClass  = {cs.IT}
}

@book{CT,
  author    = {Cover, Thomas M. and Thomas, J. A.},
  title     = {Elements of Information Theory},
  edition   = {2nd},
  year      = {2006},
  publisher = {Wiley-Interscience},
  address   = {Hoboken, NJ, USA},
  isbn      = {978-0-471-24195-9}
}

@article{Dempster1967,
  author  = {Dempster, Arthur P.},
  title   = {Upper and lower probabilities induced by a multivalued mapping},
  journal = {Annals of Mathematical Statistics},
  year    = {1967},
  volume  = {38},
  number  = {2},
  pages   = {325--339},
  doi     = {10.1007/978-3-540-44792-4_3}
}

@article{Deng_extropy2024,
  author  = {Deng, X. and Xue, S. and Jiang, W. and Zhang, X.},
  title   = {Plausibility Extropy: The Complementary Dual of Plausibility Entropy},
  journal = {IEEE Transactions on Systems, Man, and Cybernetics: Systems},
  year    = {2024},
  volume  = {54},
  number  = {11},
  pages   = {6936--6947},
  doi     = {10.1109/TSMC.2024.3444811}
}

@article{ell,
  author  = {Ellerman, D.},
  title   = {The Logic of Partitions: Introduction to the Dual of the Logic of Subsets},
  journal = {Review of Symbolic Logic},
  year    = {2010},
  volume  = {3},
  number  = {2},
  pages   = {287--350},
  doi     = {10.1017/S1755020310000018}
}

@article{Gao+,
  author  = {Gao, X. and Deng, Y.},
  title   = {The generalization negation of probability distribution and its application in target recognition based on sensor fusion},
  journal = {International Journal of Distributed Sensor Networks},
  year    = {2019},
  volume  = {15},
  number  = {5},
  doi     = {10.1177/1550147719849381}
}

@article{ho2010interplay,
  author  = {Ho, S. W. and Verd{\'u}, S.},
  title   = {On the interplay between conditional entropy and error probability},
  journal = {IEEE Transactions on Information Theory},
  year    = {2010},
  volume  = {56},
  number  = {12},
  pages   = {5930--5942},
  doi     = {10.1109/TIT.2010.2080891}
}

@inproceedings{ishida2017learning,
  author    = {Ishida, T. and Niu, G. and Hu, W. and Sugiyama, M.},
  title     = {Learning from complementary labels},
  booktitle = {Advances in Neural Information Processing Systems},
  volume    = {30},
  year      = {2017},
  url       = {https://proceedings.neurips.cc/paper/2017/hash/6364d3f5207911c50494576d3330e42d-Abstract.html}
}

@inproceedings{Kaur_comment,
  author    = {Kaur, M. and Srivastava, A.},
  title     = {Negation of a Probability Distribution: A Short Comment},
  booktitle = {2022 8th International Conference on Signal Processing and Communication (ICSC)},
  year      = {2022},
  pages     = {439--444},
  doi       = {100.1109/ICSC56524.2022.10009460},
  publisher = {IEEE}
}

@article{Kaur_markov,
  author  = {Kaur, M. and Srivastava, A.},
  title   = {A note on negation of a probability distribution},
  journal = {Soft Computing},
  year    = {2023},
  volume  = {27},
  pages   = {667--676},
  doi     = {10.1007/s00500-022-07635-2}
}

@article{kaur,
  author  = {Kaur, M. and Srivastava, A.},
  title   = {Negation of a probability distribution: An information theoretic analysis},
  journal = {Communications in Statistics - Theory and Methods},
  year    = {2023},
  volume  = {53},
  number  = {17},
  pages   = {6252--6265},
  doi     = {10.1080/03610926.2023.2242986}
}

@article{klein2022some,
  author  = {Klein, I.},
  title   = {Some technical remarks on negations of discrete probability distributions and their information loss},
  journal = {Mathematics},
  year    = {2022},
  volume  = {10},
  number  = {20},
  pages   = {3893},
  doi     = {10.3390/math10203893}
}

@book{kreinovich2018beyond,
  editor    = {Kreinovich, V. and Thach, N. N. and Trung, N. D. and Van Thanh, D.},
  title     = {Beyond Traditional Probabilistic Methods in Economics},
  year      = {2019},
  publisher = {Springer},
  address   = {Cham, Switzerland},
  series    = {Studies in Computational Intelligence},
  doi       = {10.1007/978-3-030-04200-4},
  isbn      = {978-3-030-04199-1}
}

@article{LiuDL,
  author  = {Liu, R. and Deng, Y. and Li, Z.},
  title   = {The maximum entropy negation of basic probability assignment},
  journal = {Soft Computing},
  year    = {2023},
  volume  = {27},
  pages   = {7011--7021},
  doi     = {10.1007/s00500-023-08038-7}
}

@article{Luo_matrix2020,
  author  = {Luo, Z. and Deng, Y.},
  title   = {A Matrix Method of Basic Belief Assignment's Negation in {Dempster}--{Shafer} Theory},
  journal = {IEEE Transactions on Fuzzy Systems},
  year    = {2020},
  volume  = {28},
  number  = {9},
  pages   = {2270--2276},
  doi     = {10.1109/TFUZZ.2019.2930027}
}

@book{marshall1979inequalities,
  author    = {Marshall, A. W. and Olkin, I. and Arnold, B. C.},
  title     = {Inequalities: Theory of Majorization and Its Applications},
  year      = {2010},
  edition   = {2nd},
  publisher = {Springer},
  address   = {New York},
  doi       = {10.1007/978-0-387-68276-1},
  isbn      = {978-0-387-68276-1}
}

@book{mceliece2002theory,
  author    = {McEliece, R. J.},
  title     = {The Theory of Information and Coding},
  year      = {2002},
  edition   = {2nd},
  publisher = {Cambridge University Press},
  address   = {Cambridge, UK},
  isbn      = {978-0-521-00032-1}
}

@article{pham_estimating2021,
  author  = {Pham, U. and Batyrshin, I. and Kubysheva, N. and Kosheleva, O.},
  title   = {Estimating a probability distribution corresponding to the negation of a property},
  journal = {Soft Computing},
  year    = {2021},
  volume  = {25},
  number  = {5},
  pages   = {7975--7983},
  doi     = {10.1007/s00500-021-05728-y}
}

@article{sason2016f,
  title={$f$-divergence inequalities},
  author={Sason, I. and Verd{\'u}, S.},
  journal={IEEE Transactions on Information Theory},
  volume={62},
  number={11},
  pages={5973--6006},
  year={2016},
  publisher={IEEE}
}

@book{Shafer1976,
  author    = {Shafer, G.},
  title     = {A Mathematical Theory of Evidence},
  year      = {1976},
  publisher = {Princeton University Press},
  address   = {Princeton, NJ},
  isbn      = {978-0-691-08175-5}
}

@incollection{smets1994,
  author    = {Smets, P.},
  title     = {What is {Dempster}--{Shafer}'s model?},
  booktitle = {Advances in the Dempster-Shafer Theory of Evidence},
  editor    = {Yager, R. R. and Fedrizzi, M. and Kacprzyk, J.},
  year      = {1994},
  pages     = {5--34},
  publisher = {John Wiley \& Sons},
  address   = {New York}
}

@article{smets1994_2,
  author  = {Smets, P. and Kennes, R.},
  title   = {The transferable belief model},
  journal = {Artificial Intelligence},
  year    = {1994},
  volume  = {66},
  pages   = {191--234},
  doi     = {10.1016/0004-3702(94)90026-4}
}

@article{Sun2020,
  author  = {Sun, C. and Li, S. and Deng, Y.},
  title   = {Determining weights in multi-criteria decision making based on negation of probability distribution under uncertain environment},
  journal = {Mathematics},
  year    = {2020},
  volume  = {8},
  number  = {2},
  pages   = {191},
  doi     = {10.3390/math8020191}
}

@inproceedings{szegedy2016rethinking,
  author    = {Szegedy, C. and Vanhoucke, V. and Ioffe, S. and Shlens, J. and Wojna, Z.},
  title     = {Rethinking the Inception Architecture for Computer Vision},
  booktitle = {Proceedings of the IEEE Conference on Computer Vision and Pattern Recognition (CVPR)},
  year      = {2016},
  pages     = {2818--2826},
  publisher = {IEEE}
}

@article{tanwar2023generalization,
  author  = {Tanwar, P. and Srivastava, A.},
  title   = {Generalization of negation of a probability distribution},
  journal = {International Journal of System Assurance Engineering and Management},
  year    = {2023},
  volume  = {14},
  number  = {1},
  pages   = {447--454},
  doi     = {10.1007/s13198-023-01874-8}
}

@article{TANWAR2023113557,
  author  = {Tanwar, P. and Srivastava, A.},
  title   = {Negation and redistribution with a preference - An information theoretic analysis},
  journal = {Chaos, Solitons \& Fractals},
  year    = {2023},
  volume  = {172},
  pages   = {113557},
  doi     = {10.1016/j.chaos.2023.113557}
}

@article{tinaztepe2023application,
  author  = {Tinaztepe, R.},
  title   = {An application of {Ky} {Fan} inequality: on {Kullback}--{Leibler} divergence between a probability distribution and its negation},
  journal = {Journal of Mathematical Inequalities},
  year    = {2023},
  volume  = {17},
  number  = {4},
  pages   = {1639--1646},
  doi     = {10.7153/jmi-2023-17-107}
}

@article{Wu_exp,
  author  = {Wu, Q. and Deng, Y. and Xiong, N.},
  title   = {Exponential negation of a probability distribution},
  journal = {Soft Computing},
  year    = {2022},
  volume  = {26},
  pages   = {2147--2156},
  doi     = {10.1007/s00500-021-06658-5}
}

@article{Xu+,
  author  = {Xu, S. and Hou, Y. and Deng, X. and Chen, P. and Zhou, S.},
  title   = {Logarithmic negation of basic probability assignment and its application in target recognition},
  journal = {Information},
  year    = {2022},
  volume  = {13},
  number  = {8},
  pages   = {387},
  doi     = {10.3390/info13080387}
}

@article{Yager2015,
  author  = {Yager, R. R.},
  title   = {On the Maximum Entropy Negation of a Probability Distribution},
  journal = {IEEE Transactions on Fuzzy Systems},
  year    = {2015},
  volume  = {23},
  number  = {5},
  pages   = {1899--1907},
  doi     = {10.1109/TFUZZ.2014.2374211}
}

@article{YinDD,
  author  = {Yin, L. and Deng, X. and Deng, Y.},
  title   = {The negation of a basic probability assignment},
  journal = {IEEE Transactions on Fuzzy Systems},
  year    = {2018},
  volume  = {27},
  number  = {1},
  pages   = {135--143},
  doi     = {10.1109/TFUZZ.2018.2871756}
}

@article{Zadeh1965,
  author  = {Zadeh, L. A.},
  title   = {Fuzzy sets},
  journal = {Information and Control},
  year    = {1965},
  volume  = {8},
  number  = {3},
  pages   = {338--353},
  doi     = {10.1016/S0019-9958(65)90241-X}
}

@article{Zhang,
  author  = {Zhang, J. and Liu, R. and Zhang, J. and Kang, B.},
  title   = {Extension of {Yager}'s negation of a probability distribution based on {Tsallis} entropy},
  journal = {International Journal of Intelligent Systems},
  year    = {2020},
  volume  = {35},
  number  = {1},
  pages   = {72--84},
  doi     = {10.1002/int.22198}
}

\end{document}